\documentclass[runningheads]{llncs}
\usepackage[utf8]{inputenc}
\usepackage{epsfig}
\usepackage{graphics}
\usepackage{psfrag}
\usepackage{comment}

\usepackage{amsfonts,amssymb,amsmath,latexsym,ae,aecompl,bbm,algorithm,amsthm}
\usepackage{color}

\newtheorem{observation}{Observation}

\begin{document}
\title{Efficient Dispersion on an Anonymous Ring in the Presence of Weak Byzantine Robots\thanks{The work of W. K. Moses Jr. was supported in part by a Technion fellowship. A. R. Molla is supported, in part, by DST INSPIRE Faculty Research Grant DST/INSPIRE/04/2015/002801, Govt. of India and ISI DCSW/TAC Project (file number E5412).}}

\author{Anisur Rahaman Molla\inst{1} \and
Kaushik Mondal\inst{2} \and
William K. Moses Jr.\inst{3}}
\authorrunning{Molla et al.}
%
\institute{Computer and Communication Sciences, Indian Statistical Institute, Kolkata, India.\\
\email{molla@isical.ac.in} \and
Dept. of Mathematics, Indian Institute of Technology Ropar, India\\
\email{kaushik.mondal@iitrpr.ac.in} \and
Faculty of Industrial Engineering and Management, Technion - Israel Institute of Technology, Haifa, Israel.
\email{wkmjr3@gmail.com}}

\titlerunning{Dispersion on an Anonymous Ring in the Presence of Byzantine Robots}

\maketitle
\begin{abstract}
The problem of dispersion of mobile robots on a graph asks that $n$ robots initially placed arbitrarily on the nodes of an $n$-node anonymous graph, autonomously move to reach a final configuration where exactly each node has at most one robot on it. This problem is of significant interest due to its relationship to other fundamental robot coordination problems, such as exploration, scattering, load balancing, relocation of self-driving electric cars to recharge stations, etc. The robots have unique IDs, typically in the range $[1,poly(n)]$ and limited memory, whereas the graph is anonymous, i.e., the nodes do not have identifiers. The objective is to simultaneously minimize two performance metrics: (i) time to achieve dispersion and (ii) memory requirement at each robot. This problem has been relatively well-studied when robots are non-faulty.

In this paper, we introduce the notion of Byzantine faults to this problem, i.e., we formalize the problem of dispersion in the presence of up to $f$ Byzantine robots. We then study the problem on a ring while simultaneously optimizing the time complexity of algorithms and the memory requirement per robot. Specifically, we design deterministic algorithms that attempt to match the time lower bound ($\Omega(n)$ rounds) and memory lower bound ($\Omega(\log n)$ bits per robot).

Our main result is a deterministic algorithm that is both time and memory optimal, i.e., $O(n)$ rounds and $O(\log n)$ bits of memory required per robot, subject to certain constraints. We subsequently provide results that require less assumptions but are either only time or memory optimal but not both. We also provide a primitive, utilized often, that takes robots initially gathered at a node of the ring and disperses them in a time and memory optimal manner without additional assumptions required.
\end{abstract}

\section{Introduction}
\label{sec:intro}

What was once the purview of science fiction has gradually descended into the realm of reality. Once only seen in movies, today we see robots used in every facet of life, from assembling parts in factories to robot-assisted surgery. With the ever decreasing costs of hardware, it has become possible to deploy large numbers of robots for various tasks previously unheard of. For example, after the Fukushima incident, large numbers of complex robots were deployed for manifold jobs ranging from investigating radiation levels to helping with the clearing of debris~\cite{B18}. Effective collaboration between robots is also a rising trend among modern day tasks, such as in the case of self-driving cars. Tesla and others are investing a lot of time and money into developing algorithms for these computational entities to effectively communicate with each other in a decentralized manner in order to navigate roads safely. However, even as the tasks themselves grow varied, an important theme common to all the research is a desire for fast algorithms that simultaneously do not require the robots to have much memory.\footnote{Even though the cost of memory is decreasing day by day, a mobile robot may need to run multiple tasks in parallel, each of which adds to the memory requirement of the robot.}

The formal study of independent computational agents and their interactions is both a deep and broad area of research, spanning fields such as population protocols~\cite{AADFP06}, mobile robots~\cite{PRT11}, and programmable matter~\cite{DHRS19}
among others. The specific model of mobile robots on a graph is used to capture the abstraction of agents, limited in their movement capabilities and communication abilities, but free to move in a fixed space. Within this area, problems that are studied take on the form of either having the robots work together to find something in the graph (e.g., exploration~\cite{Bampas:2009,Cohen:2008,Das13,Dereniowski:2015,Fraigniaud:2005,MencPU17}, treasure hunting~\cite{MP15})
or form a certain configuration (e.g., gathering~\cite{CFPS12,CP02,DKLHPW11,P07}, scattering~\cite{Barriere2009,ElorB11,Poudel18,Shibata:2016}, pattern formation \cite{SY99}, convergence \cite{CP04}).

Dispersion is one such problem of the latter category. Introduced in this setting by Augustine and Moses Jr.~\cite{Augustine:2018}, it asks the following question. Given $n$ robots initially placed arbitrarily on an $n$ node graph, devise an algorithm such that the robots reach a configuration where exactly one robot is present on each node. The original paper looked at the trade-offs between time taken to reach this configuration and the memory required by each robot. Subsequent papers~\cite{Kshemkalyani,KMS2019,KMS2020-ICDCN,KMS2020,tamc19} have expanded the scope of this problem, but have always maintained this focus on time and memory efficiency.

However, none of these previous works consider faulty robots. Thus, several natural questions arise: is dispersion possible if there are faulty robots? Specifically, if the faultiness is that of Byzantine faults, which is considered to be the stronger notion among the faults.\footnote{There are mainly two types of faults-- one is ``crash fault'' which means that once a robot crashes, it is dead and will not be active again thereafter; another one is ``Byzantine fault'' which means that a robot is alive throughout and may behave maliciously. Note that the Byzantine fault subsumes the crash fault.} Furthermore, if dispersion is possible, how do the Byzantine robots influence the complexities of the algorithms? In this paper, we answer these questions, showing that dispersion in the presence of Byzantine faults is indeed possible and presenting efficient solutions for Byzantine dispersion on a ring.

The best known algorithm for dispersion on a ring ``without faulty robots'' is quite straightforward, and takes $O(n)$ rounds and $O(\log n)$ bits of memory per robot \cite{Augustine:2018}. This algorithm and the other non-faulty dispersion algorithms do not apply readily when there are Byzantine robots. In fact, the most commonly used techniques for these kinds of problems related to mobile robots on graphs (e.g., dispersion, exploration, scattering) are based on depth first search (DFS) traversals or breadth first search (BFS) traversals--which will not work immediately in the presence of Byzantine robots. The main reason is that it is difficult for a robot to distinguish between a non-Byzantine and a Byzantine robot. Furthermore, unless a robot has enough memory, it cannot remember all the robots (assuming they have unique IDs). Furthermore, the Byzantine robots know the deterministic algorithm and the positions/states of all the non-Byzantine robots. So the Byzantine robots can always occupy the empty nodes where a non-Byzantine robot is supposed to settle in a particular round, and thus a non-Byzantine robot may never get an empty node. Even if a (non-Byzantine) robot can memorize all the robots, it might take a long time to find its correct place (node) in the graph. Sometimes a non-Byzantine robot may guess that a settled robot is a Byzantine robot and settle at the same node--which might result in incorrect dispersion if the guess is wrong.  So either the dispersion is wrong, or it takes a long time, or the robots require large memory. We explore all these difficulties in this paper.


\subsection{Model}\label{subsec:model}
Consider a ring with $n$ nodes. The ring is anonymous in the sense that the nodes are indistinguishable (they do not have identifiers), but the ports have unique labels. Each node in the ring has two ports that correspond to the edges from it, with unique port numbers assigned to each port. Note that an edge between adjacent nodes may have different port numbers assigned to it. Consider $n$ robots initially placed on arbitrary nodes. When all $n$ robots are initially placed on the same node, we call the ring a $\textit{rooted}$ ring.

Robots are distinguishable, i.e., each robot has a unique identifier assigned to it from the range $[1,n^c]$, where $c>1$ is a constant, unless otherwise stated. Two robots co-located on the same node can communicate with each other. One way to understand this communication between robots is as follows. Each robot has two types of memory-- one is {\em exposed}, and the other is {\em unexposed}. Any information present in the exposed memory can be read/scanned by the other co-located robots. Information in the unexposed memory is hidden from the others. Since we assumed that a robot cannot change its ID, each robot's unique identifier would be considered to be permanently stored in its exposed memory and not be changeable.

If a robot moves from one node to an adjacent node, it is aware of both port numbers assigned to the edge through which it passed. We also note that a robot present on a node can observe the port through which another robot enters that node. We recursively define the notion of \emph{clockwise} and \emph{counter-clockwise} directions for each robot. Consider that a given robot moves from node $u$ to node $v$ through $u$'s clockwise edge. Now, for node $v$, the edge $\langle v,u \rangle$ is its counter-clockwise edge and the other edge is its clockwise edge. When a robot first starts the algorithm (before moving anywhere), it assigns the directions of clockwise and counter-clockwise as follows. It observes the port numbers of the node it is initially placed on. Denote the edge with the lower port number as clockwise and the other edge as counter-clockwise.\footnote{In the algorithms, we mention a robot \textit{resets its sense of direction}, i.e., it resets its notion of clockwise and counter-clockwise. That refers to the robot performing this check again and redefining clockwise and counter-clockwise accordingly.} Notice that each robot has its own sense of clockwise and counter-clockwise, but two robots may not agree on this sense.

We adapt the definition of a weak Byzantine robot from~\cite{DPP14}. A Byzantine robot may behave maliciously and arbitrarily, i.e., it may share wrong information, perform moves that are deviations from the algorithm, etc. As in~\cite{DPP14}, we assume that a Byzantine robot cannot fake its ID, i.e., it cannot communicate to a robot that its ID has a value other than the one initially assigned to it. Note that the exposed memory of a Byzantine robot can be read by all other robots co-located with it. Among the $n$ robots, up to $f$ of them are considered to be Byzantine. Some of our algorithms can afford up to $n-1$ Byzantine robots among $n$ robots. Moreover, our algorithms work without knowing the number of Byzantine robots in the system. That is, each robot knows the value of $n$, but need not know the value of $f$, unless otherwise stated. The Byzantine robots may work together to thwart the algorithm. Imagine an adversary coordinates and controls all the Byzantine robots. The adversary has the knowledge of the algorithm, knows the non-Byzantine robots and their states throughout the algorithm. A non-Byzantine robot cannot distinguish between a Byzantine and a non-Byzantine robots in the beginning. By designing our algorithms in the face of such an adversary, we ensure that they are robust to all manner of deviations from these Byzantine robots.

We consider a synchronous system, where in each round a robot performs the following tasks in order: (i) Robots that are co-located at the same node instantaneously and simultaneously read each other's exposed memory.  Robots may perform some local computation and may update information in their unexposed memory. (ii) Robots update their exposed memory as needed. (iii) Each robot either stays at the same node or moves to another node.

Note that task (i) of each round seems to require each robot to have a large memory. However, we have written it this way for ease of understanding. We ensure that our algorithms can simulate task (i) using the memory we allocate to run those algorithms.  However, we do restrict the Byzantine robots to not change their exposed memory before we reach task (ii) of a round. We assume that all robots are initially awake and can engage in the algorithms from the beginning.

Considering the fact that a Byzantine robot can settle at any node (we have no control on it), let us now formally define the problem of dispersion on a ring in the presence of Byzantine robots. Let us call this problem {\em Byzantine dispersion}.
\begin{definition}[Byzantine Dispersion]
\label{def:dispersion}
Given $n$ robots, up to $f$ of which are Byzantine, initially placed arbitrarily on a ring of $n$ nodes, the robots re-position themselves autonomously to reach a configuration where each node has at most one non-Byzantine robot on it and terminate.
\end{definition}
The problem of dispersion on a ``rooted ring'' is a variation of the above problem statement where all $n$ robots are initially located on the same node.

\subsection{Our Contributions}\label{subsec:our-contrib}
In this paper, we introduce the notion of Byzantine robots to the problem of dispersion of mobile robots on graphs. In the context of mobile robots on a graph, previously only the problem of gathering has been extended to the setting where Byzantine robots are present. Those previous results focused on whether gathering can be solved and how many non-Byzantine robots are required to solve the problem. In contrast to those previous results, we focus on the efficiency of solutions instead of just answering the question of whether dispersion can be achieved or not. More specifically, we care about the time and memory efficiency of solutions and seek to minimize them.

Augustine and Moses Jr.~\cite{Augustine:2018} showed a lower bound of $\Omega(\log n)$ bits of memory per robot in order for a deterministic algorithm to achieve dispersion. For a ring, it is easy to see that a lower bound on time complexity is $\Omega(n)$ rounds (since the diameter is $n/2$). Thus, our goal is to develop an algorithm that solves Byzantine dispersion with time and memory complexities that match these lower bounds.

Against this backdrop, we present one procedure and three deterministic algorithms for Byzantine dispersion. Our main result is a time and memory optimal algorithm for Byzantine dispersion on rings.

We first develop an important building block used in subsequent algorithms, the procedure \textsc{Rooted-Ring-Dispersion}. It achieves Byzantine dispersion on a rooted ring in at most $n-1$ rounds and requires each robot to have $O(\log n)$ bits of memory. This procedure allows $k\leq n$ co-located robots (where all non-Byzantine robots are present) with unique IDs taken from any range to achieve Byzantine dispersion even when $n$ is unknown and $f$ can be as large as $k-1$.

Our first and most important contribution is a time and memory optimal algorithm, \textsc{Opt-Ring-Dispersion}, which solves Byzantine dispersion on a ring in $O(n)$ rounds and uses $O(\log n)$ bits of memory per robot. The algorithm relies on the following four assumptions: (i) the ID space of robots is restricted to the range $[1,n]$, (ii) the upper bound on the number of Byzantine robots $f$ is known to the robots, (iii) $f$ is restricted to $f \leq \lfloor (n-4)/17 \rfloor$, and (iv) \textit{follow} primitive holds, i.e., one robot may follow another robot (refer to Section~\ref{sec:opt-ring-alg} for more details).

Our second contribution is a memory optimal algorithm for Byzantine dispersion on a ring, \textsc{Mem-Opt-Ring-Dispersion}, which requires less assumptions--it only requires the ID space of the robots to be restricted to $[1,n]$. \textsc{Mem-Opt-Ring-Dispersion} takes $O(n^2)$ rounds and uses only $O(\log n)$ bits of memory per robot.

Our final contribution is a time optimal algorithm, \textsc{Time-Opt-Ring-Dispersion}, which requires no assumptions at all. It takes $n$ rounds and uses $O(n \log n)$ bits of memory per robot. It should be noted that this algorithm has a very tight running time. Our results are summarized in Table~\ref{table:results}.

\begin{table*}[ht]
	\caption{Our results for Byzantine dispersion of $n$ robots on an $n$ node ring in the presence of at most $f$ Byzantine robots. The first column gives the algorithm's name. The second and third give its time and memory complexity, respectively. The final column lists the assumptions necessary for the algorithm to work. The possible assumptions are: (i) restricted ID space - each robot's unique ID is taken from $[1,n]$, (ii) $f\leq \lfloor (n-4)/17 \rfloor$ known - each robot must know the value of $f$ and $f$ is restricted to values $\leq \lfloor (n-4)/17 \rfloor$, and (iii) \textit{follow} primitive holds, where one robot may follow another.}
	\centering 
		\resizebox{1.0\columnwidth}{!}{%
	\begin{tabular}{|c|c|c|c|}
		\hline
		Algorithm & Running Time & Memory Requirement  & Assumptions Required \\
		 &  (in rounds) & (bits per robot) &   \\
		\hline
		\hline
		\textsc{Mem-Opt-Ring-Dispersion} & $O(n^2)$ & $O(\log n)$ & Restricted ID space \\
		\hline
		\textsc{Time-Opt-Ring-Dispersion} & $n$ & $O(n \log n)$ & None\\
		\hline
		\textsc{Opt-Ring-Dispersion} & $O(n)$ & $O(\log n)$ & Restricted ID space, $f \leq \lfloor (n-4)/17 \rfloor$ known, \textit{follow} \\
		\hline
	\end{tabular}
		}
	\label{table:results}
\end{table*}

\subsection{Technical Difficulties and High-Level Ideas}\label{subsec:tech-diff}
We now highlight some of the key difficulties that make this problem interesting. In doing so, we provide insight into our algorithmic design choices and a high level intuition of our algorithms, though some key details are elaborated upon only in the respective sections.

A fundamental difficulty behind any algorithm for this problem is that Byzantine robots can lie about what they have seen so far. This makes relying on communication between robots risky.  One possibility (\textsc{Time-Opt-Ring-Dispersion}) is to have each robot function independent of the others, with the only real communication between two robots being to see if one of them already settled at the current node. However, this approach requires each robot to develop a way to determine if another robot is Byzantine and remember this information since we do not want two non-Byzantine robots to settle at the same node. Since $O(n)$ robots could be Byzantine, each robot requires $O(n \log n)$ bits of memory.

However, this approach of no communication breaks down when we want each robot to have $o(n \log n)$ bits of memory each. We then require some method for a robot to safely figure out where it can settle, without having to remember the IDs of all the Byzantine robots. In this paper, we have developed a useful primitive (\textsc{Rooted-Ring-Dispersion}) that allows robots with only $O(\log n)$ bits of memory to achieve Byzantine dispersion. The only catch is that all non-Byzantine robots should already be present on the same node. Thus our problem reduces to one of gathering. How best can we gather robots with limited memory?

If we cannot remember all the robots which are Byzantine, is there a technique to gather where each robot does not need to remember information about all robots all the time? One way to do this is to restrict the rounds in which each robot is allowed to move. Recall that a Byzantine robot cannot lie about its ID. If we restrict a robot with ID $x$ to only move in rounds $f(x,1), f(x,2), \ldots$, where $f(x,1)$ is a function known to all robots, and $f(x,i) \neq f(y,j)$ when $x \neq y$, then we provide a memory-lite way for a robot to identify Byzantine robots. By having robots interact with each other in smart ways and guaranteeing that by a certain round, all non-Byzantine robots have gathered, we solve Byzantine dispersion while requiring robots to only have $O(\log n)$ bits of memory each (\textsc{Mem-Opt-Ring-Dispersion}). However, the algorithm takes $O(n^2)$ rounds and requires the ID space of robots to be restricted to $[1,n]$.

The key reason the previous algorithm took so many rounds is that we restricted $f(x,i) \neq f(y,j)$ when $x \neq y$. This was to ensure that a single robot does not need to keep track of multiple Byzantine robots (and the associated IDs) in a given round.
However, when the actual value of $f$, the upper bound on the number of Byzantine robots is known, then we may be a little clever. By looking for a group of robots with at least $2f+1$ robots moving together, a robot needs only $O(\log n)$ bits of memory and it can safely follow the group because a majority of the robots in the group are non-Byzantine robots. This helps us eventually gather all non-Byzantine robots together while allowing multiple robots to move at the same time. However, in order to form this initial group of at least $2f+1$ robots, and ensure that it is the only group that is initially formed requires a bit of work, as seen in \textsc{Opt-Ring-Dispersion}.

\subsection{Related Work}\label{subsec:rel-work}
The problem of dispersing mobile robots on graphs was first introduced by Augustine and Moses Jr. \cite{Augustine:2018} and they provided solutions for various types of graphs including paths, rings, trees, and general graphs. After that, the problem was studied by several papers \cite{Kshemkalyani,KMS2019,KMS2020-ICDCN,KMS2020,tamc19} in various settings to improve the efficiency of the solutions. The best known time-memory efficient algorithm for dispersion of $k\leq n$ robots on an arbitrary $n$-node graph has time complexity of $O(\min\{m, k\Delta\}\log n)$ rounds and $O(\log n)$ bits of memory per robot \cite{KMS2019}, where $\Delta$ is the maximum degree of the graph. The paper \cite{KMS2020} studies dispersion on the grid graph and provides a $O(\sqrt{n})$ time algorithm using $O(\log n)$ memory for each robot, an optimal solution with respect to both memory and time. Randomized algorithms are presented in \cite{tamc19} where random bits are mainly used to break the memory requirement of $\Omega(\log n)$ bits per robot. The papers \cite{KMS2020-ICDCN,KMS2020} study the problem in a slightly different communication model--{\em global communication model}. In this model, a robot can communicate with all other robots in the graph, irrespective of their positions. However, each robots does not have the position information of other robots (unless they are co-located at the same node) as graph nodes are anonymous. Note that this paper and all the other papers mentioned above assume only {\em local communication}, in which a robot can only communicate with other robots present at the same node. The paper \cite{KMS2020-ICDCN} shows that the global communication doesn't help much to speed up the run time of the dispersion algorithms compared to the local communication.

Some other problems that are closely related to dispersion on graphs are exploration, scattering, gathering, etc. The problem of graph exploration has been studied extensively in the literature for specific as well as arbitrary graphs,
e.g., \cite{Bampas:2009,Cohen:2008,Das13,Dereniowski:2015,Fraigniaud:2005,MencPU17}. While some of these exploration algorithms (especially those which fit the current model) can be adapted to solve dispersion (with additional work), they, however, provide inefficient time-memory bounds; a detailed comparison is given in \cite{KMS2019}.
Another problem related to dispersion is {\em scattering} (also known as uniform-deployment) of mobile robots in a graph and is also studied by several papers, e.g., it has been studied for rings \cite{ElorB11,Shibata:2016} and grids \cite{Barriere2009,Poudel18} under different assumptions. Load balancing, where a given
load at the nodes has to be (re-)distributed among several processors (nodes) is also relevant to dispersion where robots can be used to distribute the loads. This problem has been studied in graphs, e.g., \cite{Cybenko:1989,Subramanian:1994}.
Another problem which can be used as a subroutine to solve dispersion (in some specific cases) is {\em gathering} of mobile robots on graphs. Once all robots are gathered at a single node, a DFS/BFS type traversal algorithm can be run to disperse the robots from that node. Particularly, gathering can help solve dispersion easily on paths, rings, grids, or tree like structures, e.g., see Section~\ref{sec:rooted-ring} for the ring.

While it appears that the notion of Byzantine robots is not new to the mobile robots literature in general (e.g.,~\cite{BPT09,ABCTU13,CGKKNOS16}), it appears that its usage in the context of mobile robots on a graph is fairly recent.
To the best of our knowledge, only the problem of gathering has been studied in the context of Byzantine robots in the graph setting. Specifically, Dieudonn\'e et al.~\cite{DPP14} introduced the notion of Byzantine robots to the gathering problem. The paper mainly investigates the possibility and impossibility of gathering of mobile robots on graphs in the presence of Byzantine robots. They present some possibility results under certain assumptions on the minimum number of non-Byzantine robots present. Their solution can be adapted to solve Byzantine dispersion on a ring, but gives a time-memory inefficient solution--$\Omega(n^4)$ rounds and $\Omega(n \log n)$ bits of memory per robot.\footnote{Specifically, by Theorem~3.6 in~\cite{DPP14}, their algorithm runs in $4n^4 \cdot P(n, |\lambda|)$ rounds, where $\lambda$ is the largest label of a non-Byzantine robot and $P(n, |\lambda|)$ is a polynomial in the two variables $n$ and $\lambda$. Since $P(n, |\lambda|) = \Omega(1)$, their algorithm takes at least $\Omega(n^4)$ rounds to gather robots at a node. Subsequently utilizing the procedure~\textsc{Rooted-Ring-Dispersion}, developed in this paper, allows us to achieve Byzantine dispersion without further increasing the asymptotic time complexity. Furthermore, each robot is required to maintain a blacklist of possibly $O(n)$ Byzantine robots in their memory, requiring $O(n \log n)$ bits of memory.}
There are some follow-up papers on Byzantine gathering, mostly focused on the feasibility of the solution, e.g., \cite{BDD16,BDL18,DFLMS19,HTNOI20}, which can be adapted to solve dispersion on a ring, but result in solutions which are time-memory inefficient. Some other remarkable results on faulty robots in different models studies gathering problem \cite{DLM15,TOI18}, but in different model.

There has been previous literature related to faulty robots where faultiness manifests in the form of crash faults. In gathering,~\cite{CDLP16} and~\cite{DFLMS19} look at gathering in the presence of crash faults. For graph exploration, there is not much work done under faulty robots. There is some work on exploration with faulty tokens~\cite{DP12} and exploration on a graph with faulty edges~\cite{CP16}. There is also some work done on graphs that are ``dangerous"~\cite{FKMS12,MS19}, i.e., the nodes/links are faulty in some way.

\subsection{Paper Organization}\label{subsec:paper-org}
In Section~\ref{sec:rooted-ring}, we develop a procedure for the rooted ring, which is used as a building block in subsequent algorithms. In Section~\ref{sec:opt-ring-alg}, we present our main result, a time and memory optimal algorithm for Byzantine dispersion on the ring. In Section~\ref{sec:other-ring-alg}, we present algorithms which are either only time optimal or only memory optimal, but require less assumptions than our main result. Finally, we present our conclusions and future directions of research in Section~\ref{sec:conclusion}.

\section{Building Block}\label{sec:rooted-ring}
In this section, we present a procedure to achieve Byzantine dispersion on the rooted ring, i.e., when all robots start at the same node initially. This procedure also works when some subset of the robots that includes all non-Byzantine robots are co-located initially. This procedure is subsequently used in our algorithms. The procedure, \textsc{Rooted-Ring-Dispersion}, is a simple one that finishes in $O(n)$ rounds and requires each robot to have $O(\log n)$ bits of memory. It can handle any number of Byzantine robots, does not require robots to know the value of $n$, and works even when robots have unique IDs taken from some arbitrarily large range (but still polynomial in $n$). The procedure works as follows.

In the first round, all co-located robots communicate with each other and determine their position in the total order of IDs as follows. In order for a robot to identify its position in the total order, it maintains a $\log n$ bit counter, initialized to $1$, which it increments for every robot it sees in this round with a lower ID.\footnote{Note that it is not necessary for a robot to know the value of $n$ in order to maintain a counter using $\log n$ bits of memory given that the robot's total memory is $c \log n$ bits of memory, where $c$ is a sufficiently large constant.}

Still in the first round, each robot then resets its sense of direction so that all robots have the same sense of clockwise direction.

Now, a robot whose position in the total order is $i$, moves $i-1$ steps in the clockwise direction on the ring (in $i-1$ rounds) and settles down at that $(i-1)^{th}$ node and terminates the algorithm. Thus, Byzantine dispersion is achieved in at most $n-1$ rounds.

\begin{theorem}
Consider an $n$-node ring with $k \leq n$ robots placed on a single node such that all non-Byzantine robots are present among the $k$ robots and at most $f$ of them are Byzantine, $f \leq k-1$. Each robot has a unique ID, $O(\log n)$ bits of memory, and does not have knowledge of the values of $n$ and $f$. Then Byzantine dispersion can be achieved in at most $n-1$ rounds.
\end{theorem}

\section{Time and Memory Optimal Algorithm}\label{sec:opt-ring-alg}
We now describe an algorithm, \textsc{Opt-Ring-Dispersion}, that achieves Byzantine dispersion on a ring in optimal time ($O(n)$ rounds) using optimal memory ($O(\log n)$ bits) given that robots' unique IDs are restricted to the range $[1,n]$, robots know the value of $f$, and $f \leq \lfloor (n-4)/17 \rfloor$. We also require the following assumption which we call the \textit{follow} primitive. When two robots $A$ and $B$ are co-located on the same node, one robot (say $A$) can follow the movement of the other robot $B$. It is important to note that even if $B$ is a Byzantine robot, $A$ can follow the movement of $B$ when they are co-located.

The algorithm has two \textbf{Phases}. In \textbf{Phase~1}, all non-Byzantine robots gather at a node in $O(n)$ rounds. In \textbf{Phase~2}, these gathered robots perform dispersion in an additional $n$ rounds.

 \textbf{Phase 1:} The first phase of the algorithm is further subdivided into three \textbf{Sub-phases}. We first present intuition for these sub-phases, and then delve into details subsequently. The first sub-phase consists of rounds $1$ to $n$ and is used to aggregate the non-Byzantine robots into at most $f+1$ groups of robots at different nodes. The second sub-phase consists of rounds $n+1$ to $2n+1$ and has these groups move in a way so that a sufficient number of non-Byzantine robots gather together, i.e., at least $f+1$ non-Byzantine robots. The final sub-phase consists of rounds $2n+2$ to $3n+1$ and is used by these at least $f+1$ non-Byzantine robots to move around the ring and collect the remaining non-Byzantine robots.

\textbf{Sub-phase 1:} For the first $n$ rounds, those robots with ID $\in [1,f+1]$ move along the ring in the clockwise direction. Robots with IDs $\notin [1,f+1]$, once they see a robot with an ID $\in [1,f+1]$, follow that robot until the end of round $n$. If multiple robots from $\in [1,f+1]$ are seen at the same time, one is chosen arbitrarily and followed.

\textbf{Sub-phase 2:} For the next $n+1$ rounds, we describe the strategy for each robot depending on who and how many other robots are co-located with them at the start of round $n+1$. Initially, all robots reset their sense of direction so that all co-located robots have the same sense of clockwise and counter-clockwise. Call the robot with ID $1$, $R_1$. The algorithm instructs $R_1$ not to move for these $n+1$ rounds. All robots co-located with $R_1$ follow it.\footnote{Notice that we say that other co-located robots are to follow $R_1$, instead of just staying put. This is to ensure that all robots initially co-located with $R_1$ continue to stay with $R_1$, even if $R_1$ is a Byzantine robot and moves around during the $n+1$ rounds.} The goal of this second sub-phase is to have a sufficient number of non-Byzantine robots find and subsequently follow $R_1$.

We now look at how the other groups of robots not containing $R_1$ move in this sub-phase. If at some node, there are less than four robots, those robots do not move in this sub-phase.

All remaining robots in the ring are present in groups of size 4 or more. Each group is divided into four subgroups $\lbrace G_{LL}, G_{LU}, G_{UL}, G_{UU} \rbrace$, as described below, that move until either the sub-phase ends or they come into contact with robot $R_1$, in which case they subsequently follow $R_1$ until the end of this sub-phase. Each group $G$ is first divided into two subgroups: $\lfloor |G|/2 \rfloor$ of the robots with the lowest IDs form $G_L$ and the remaining form $G_U$. Again $\lfloor |G_L|/2 \rfloor$ of the lowest ID robots of $G_L$ form the subgroup $G_{LL}$ and the remaining form $G_{LU}$. From round $n+1$ to  $2n$, robots in $G_{LU}$ move in the clockwise direction. Robots in $G_{LL}$ do nothing in round $n+1$, but from round $n+2$ to round $2n+1$, robots in $G_{LL}$ move in the clockwise direction. Similarly, $\lfloor |G_U|/2 \rfloor$ of the lowest ID robots of $G_U$ form the group $G_{UL}$ and the remaining robots form $G_{UU}$. They mimic the strategies of $G_{LL}$ and $G_{LU}$ respectively but for the counter-clockwise direction. By the end of this sub-phase, for each of these groups, at least one of the four subgroups comes into contact with $R_1$. 

\textbf{Sub-phase 3:} The third sub-phase, from round $2n+2$ to $3n+1$ sees those robots which were co-located with $R_1$ at the end of round $2n+1$ move clockwise for $n$ rounds. Each robot not co-located with $R_1$ at the end of round $2n+1$, does not move from its node until a group of at least $f+1$ robots arrive at its node and claim to be robots co-located with $R_1$ at the end of round $(2n+1)$. Upon arrival of this group, $X$ does the following. $X$ observes which port the majority of them entered the node through and sets the remaining port as clockwise. $X$ subsequently moves clockwise until the end of round $(3n+1)$. At the end of round $(3n+1)$, all non-Byzantine robots are gathered.\footnote{Possibly some Byzantine robots may also be gathered as well, but the presence of these robots does not cause problems as the subsequent procedure, \textsc{Rooted-Ring-Dispersion} is correct even in the presence of $n-1$ Byzantine robots.}

\textbf{Phase~2:} Finally, in the second phase of the algorithm, the procedure \textsc{Rooted-Ring-Dispersion} is called by these gathered robots and Byzantine dispersion is achieved in an additional $n$ rounds.

\begin{theorem}
Consider an $n$ node ring with $n$ robots initially arbitrarily placed on it. Each robot has a unique ID in $[1,n]$, $O(\log n)$ bits of memory, and knows the value of $f$, the upper bound on Byzantine robots. When $f\leq \lfloor (n-4)/17 \rfloor$, the deterministic algorithm \textsc{Opt-Ring-Dispersion} achieves Byzantine dispersion in $O(n)$ rounds.
\end{theorem}
\begin{proof}
The time and memory complexities are obvious from the algorithm. We now prove correctness, i.e., Byzantine dispersion is achieved in $O(n)$ rounds.

The first $n$ rounds (in Sub-phase~1) ensure that all non-Byzantine robots will be partitioned into at most $f+1$ groups.

There can be at most $f$ groups of 3 robots, hence at most $3f$ non-Byzantine robots do not move from round $n+1$ to round $2n+1$. From the remaining $n-3f$ robots, at least $(n - 6f)/4$ of them should meet and subsequently follow $R_1$ in these $n+1$ rounds, if they act according to the algorithm. This occurs, regardless of whatever $R_1$ chooses to do. We prove this below.

Consider one group of $\geq 4$ robots $G$ at the start of round $n+1$. Since robots in $G_L$ and in $G_U$ are visiting all the nodes of the ring from opposite directions, at some point robots from one of the two groups must encounter $R_1$. Without loss of generality, let the robots in $G_U$ encounter $R_1$. Recall that we group robots in $G_U$ into two subgroups $G_{UL}$ and $G_{UU}$. This is to ensure that at least one of the two subgroups encounters $R_1$ in case $R_1$ moves in the opposite direction to these two groups, since the groups always occupy consecutive nodes in the ring. Hence, at least $\lfloor (\lfloor |G|/2 \rfloor)/2 \rfloor$ robots from $G$ meet $R_1$.

If there is only one such group, then at least $\lfloor (\lfloor |G|/2 \rfloor)/2 \rfloor \geq  \lfloor (\lfloor (n - 3f)/2 \rfloor)/2 \rfloor \geq (n - 3f - 6)/4$ robots meet with $R_1$. However, as the number of such groups increases, the number of robots that eventually meet with $R_1$ decreases. There can be at most $f$ such groups of robots other than the group containing $R_1$. Call the groups $G_1$ to $G_f$ and recall that $\sum_{i=1}^f |G_i| \geq n - 3f$. Let $G'$ be the set of robots that eventually meet up with and subsequently follow $R_1$. We see that $|G'| \geq \sum_{i=1}^f \lfloor (\lfloor |G_i|/2 \rfloor)/2 \rfloor
\geq \sum_{i=1}^f (|G_i| - 6) /4
\geq (n-3f - 6f)/4
= (n-9f)/4$.

Since there can be up to $f$ Byzantine robots in these groups, we know that of the robots in $G'$, at least $(n-9f)/4 - f = (n-13f)/4$ of them are non-Byzantine.

According to our algorithm, this group of at least $(n-13f)/4$ robots move across the ring for the next $n$ rounds to gather all the remaining non-Byzantine robots. Once met, the remaining non-Byzantine robots simply verify that there are at least $f+1$ robots in this group and subsequently follow this group as described in the algorithm. This is the case when $f \leq \lfloor (n-4)/17 \rfloor$.

Thus, at the end of round $3n+1$, all non-Byzantine robots are gathered at the same node. A subsequent call to the procedure \textsc{Rooted-Ring-Dispersion} guarantees that Byzantine dispersion is achieved in an additional $n$ rounds. Since robots keep track of the current round, they terminate at the end of round $4n+1$ only after the desired configuration of robots on nodes is reached.
\end{proof}

\section{Time or Memory Optimal Algorithms with Reduced Assumptions}\label{sec:other-ring-alg}
In the following section, we present two algorithms, a memory optimal algorithm and a time optimal algorithm, which solve Byzantine dispersion on a ring requiring less assumptions than the algorithm presented in Section~\ref{sec:opt-ring-alg}. Both the algorithms tolerate up to $n-1$ Byzantine robots, i.e., $f\leq n-1$.

\subsection{Algorithm with Optimal Memory Complexity}\label{subsec:mem-opt-ring-alg}
We describe below an algorithm, \textsc{Mem-Opt-Ring-Dispersion}, that achieves Byzantine dispersion of $n$ robots, when up to $f$ of them are Byzantine, on an $n$ node ring in $O(n^2)$ rounds requiring each robot to have $O(\log n)$ bits of memory, when robots' unique IDs are restricted to the range $[1,n]$. The algorithm takes $n$ as input i.e., $n$ is known to the robots, but does not need to know $f$. Intuitively, the algorithm first gathers the robots in $n^2$ rounds, then uses an additional $n-1$ rounds to disperse them.

The algorithm works as follows. Define stage $i$ as consisting of the rounds from $(i-1)n + 1$ to $in$. In stage $i$, the robot with ID $i$ moves clockwise for $n$ rounds. Any other robot $x$ does nothing until it comes into contact with robot $i$.\footnote{This could happen at the beginning of the stage, if the two robots are co-located on the same node.} On seeing robot $i$, $x$ communicates with $i$ to see which port $i$ will move through. Subsequently, $x$ moves in that direction until the end of round $in$. After $n$ such stages, all the non-Byzantine robots are gathered at some node. Subsequently, in round $n^2 + 1$, the algorithm calls procedure~\textsc{Rooted-Ring-Dispersion} to achieve Byzantine dispersion in an additional $n-1$ rounds.

\begin{theorem}
Consider an $n$ node ring with $n$ robots initially arbitrarily placed on it. Each robot has a unique ID in $[1,n]$, $O(\log n)$ bits of memory, and there are at most $f$ Byzantine robots, $f \leq n-1$. The deterministic algorithm \textsc{Mem-Opt-Ring-Dispersion} achieves Byzantine dispersion in $O(n^2)$ rounds.
\end{theorem}

\begin{proof}
It is easy to see that all non-Byzantine robots terminate in $n^2 + n -1 = O(n^2)$ rounds. It is also easy to see that no robot requires more than $O(\log n)$ bits in order to maintain a round counter and also to identify its position in the total order of IDs of robots co-located with it. In order to argue correctness, we first show that at the end of stage $n$, all non-Byzantine robots are gathered together.

Among the robots with IDs in $[1,f+1]$, there is at least one non-Byzantine robot, say with ID $k$. During stage $k$, robot $k$ ensures that all non-Byzantine robots are gathered. Once gathered, non-Byzantine robots will always move together in subsequent stages and will remain gathered.

Once they are gathered, procedure~\textsc{Rooted-Ring-Dispersion} ensures that Byzantine dispersion is achieved.
\end{proof}

It is clear that if the value of $f$ is known to the robots ahead of time, they can run the algorithm for exactly $f+1$ stages and subsequently disperse. Thus we have the following corollary.

\begin{corollary}
Consider an $n$ node ring with $n$ robots initially arbitrarily placed on it. Each robot has a unique ID in $[1,n]$, $O(\log n)$ bits of memory, and there are at most $f$ Byzantine robots, $f \leq n-1$. When the value of $f$ is known to the robots, there exists a deterministic algorithm that achieves Byzantine dispersion in $O(fn)$ rounds.
\end{corollary}

It is possible to remove the ID space restriction if an upper bound, say $U$, on the range of IDs is known. By having robots execute $U$ phases of gathering, each taking $n$ rounds, it is guaranteed that all robots will gather in $O(Un)$ rounds, after which procedure~\textsc{Rooted-Ring-Dispersion} may be called.

\begin{corollary}
Consider an $n$ node ring with $n$ robots initially arbitrarily placed on it. Each robot has a unique ID in the range $[1,U]$, $O(\log n)$ bits of memory, and there are at most $f$ Byzantine robots, $f \leq n-1$. There exists a deterministic algorithm that achieves Byzantine dispersion in $O(Un)$ rounds.
\end{corollary}

\subsection{Time Optimal Algorithm}\label{subsec:time-opt-ring-alg}
We now describe an algorithm, \textsc{Time-Opt-Ring-Dispersion}, that achieves Byzantine dispersion of $n$ robots, when up to $f\leq n-1$ of them are Byzantine, on an $n$ node ring in $O(n)$ rounds using $O(n \log n)$ bits of memory per robot. Again the robots need not know the number of Byzantine robots in the system.

Each robot, having $O(n\log n)$ bits of memory, can remember all the $n$ robots. It helps to detect Byzantine robot, particularly when a robot deviates the algorithm. Intuitively, this algorithm is easy to explain, though the details are involved. Each robot $r$ moves clockwise in the ring until it finds a node to be settled down at.
There are two conditions that $r$ checks at a node $v$ before deciding to settle down. One condition is that among the robots on $v$ that claim to already be settled there, there exists a robot that is not known by $r$ to be Byzantine.
The second condition is that among all the robots currently at the node $v$, there exists a robot with ID lower than $r$'s that intends to settle and is not known by $r$ to be Byzantine. If either check succeeds, then $r$ does not settle down at $v$. Else $r$ settles down.
The checking on the second condition involves non-trivial computations. The detailed algorithm is given below.

Each robot $r$ maintains an array $A_r$ of size $n+1$, where $A_r[k]$ contains the ID(s) of the settled robot(s) that $r$ encountered in round $k\geq 1$. Note that the total number of settled robots in some node can be more than one, as Byzantine robots may settle with a non-Byzantine robot.  $A_r[0]$ is used to represent whether $r$ is settled on the current node or not. Initially $A_r[0]=0$. Robot $r$ sets $A_r[0] = 1$ when on some node $v$ in order to claim that it is settled on $v$.

Let the robot $r$ impose the local naming convention $\{v_1, v_2,...,v_n\}$ on the set of all nodes it may see in the ring, where the node it is initially placed on is $v_1$, then the next node it moves to is $v_2$, and so on. In any round $k$, let $G'^k_r=\{s_1,s_2,...,s_p\}$ be the set of $p$ already settled robots at $v_k$ at the beginning of round $k$. Let $G^k_r$ be the group of robots on node $v_k$ in round $k$, excluding those robots in $G'^k_r$. If there is no settled robot at $v_k$ by the start of round $k$, then $G'^k_r$ is empty. If there are no robots on $v_k$ in round $k$, excluding those in $G'^k_r$, then $G^k_r$ is empty.

In any round $k$, let $M^k_r$ be the the set of robots whose IDs were written in $A_r[1], A_r[2],...,A_r[k-1]$. Define $B^k_r = M^k_r \bigcap (G'^k_r \bigcup G^k_r)$ and let its complement be denoted by $B^{Ck}_r$. Intuitively, $B^k_r$ represents robots that $r$ has identified as acting in a Byzantine manner by claiming to settle at a previous node and now are present at node $v_k$.

We will shortly describe how robot $r$ determines whether it will settle down at a node $v_k$ in round $k$. For now, it is sufficient to note that the decision of $r$ to settle down at node $v_k$ in round $k$ is a result of a computation whose input is $r$'s memory and the memories of other robots co-located with it on $v_k$ at that time. This is in fact true of all the robots in $G^k_r$ that may possibly decide to settle down at the node. All these memories can be read by all robots co-located at the node in that round. Thus, robot $r$ can calculate the set $S^k_r \subset G^k_r \bigcup G'^k_r$ of robots that will settle down at node $v_k$ and add it to $A_r[k]$.

We now describe the procedure for $r$ in each round $1 \leq k \leq n$ on node $v_k$ until $r$ settles down.
\begin{enumerate}
\item\label{item:line 1} Initialize $S^k_r= G'^k_r\setminus B^k_r$. Robot $r$ performs the following check. If $S^k_r$ is not empty, then $r$ does not settle at $v_k$.

\item \label{item:line 2} Robot $r$ performs the following calculation to iteratively determine the subset of robots from $G^k_r$ that may decide to settle at the node. Consider all robots $s$ such that $s \in G^k_r$. Order these robots in ascending order of their ID. For each robot $s$ in this list from smallest ID to largest ID, do the following. (Note that $G^k_s = G^k_r$ and $G'^k_s = G'^k_r$.)
\begin{enumerate}
    \item If there does not exist a robot $t$ such that $t \in G'^k_s$ and $t \notin B^k_s$, then proceed to the next step. Else, move to the next robot in the list.
    \item If there does not exist a robot $t$ such that $t \in S^k_r$ and $t \notin B^k_s$, then add $s$ to $S^k_r$. 
\end{enumerate}
Now $r$ performs the following check. If there exists a robot $s$ such that the ID of $s$ is less than that of $r$, $s \in S^k_r$, and $s \notin B^k_r$, then $r$ does not settle at $v_k$.

\item \label{item:line 3} If neither of the above two checks are satisfied, then $r$ settles at $v_k$.

\item \label{item:line 4} If $r$ does not settle at $v_k$, it writes the robot IDs of $B^{Ck}_r \bigcap S^k_r$ in $A_r[k]$ and moves clockwise through an edge.
\end{enumerate}

Once $r$ settles down at a node, it waits until the end of round $n$ and then terminates.\footnote{If $r$ terminates prior to the end of round $n$, it becomes invisible to other robots. Thus, there is the risk of another non-Byzantine robot settling at the same node as $r$ if $r$ terminates early.} The following theorem captures the properties of the algorithm.

\begin{theorem}\label{theorem:nrounds}
Consider an $n$ node ring with $n$ robots initially arbitrarily placed on it, up to $f \leq n-1$ of which are Byzantine in nature. Each robot has a unique ID and knows the value of $n$. Algorithm \textsc{Time-Opt-Ring-Dispersion} solves Byzantine dispersion in $n$ rounds and requires each robot to have $O(n \log n)$ bits of memory.
\end{theorem}

In order to prove the theorem, we first make an observation and prove a few useful lemmas.
\begin{observation}
When $r$ is a non-Byzantine robot, for $1 \leq k \leq n$, the robots in $B^k_r$ are Byzantine.
\end{observation}
\begin{proof}
Let robot $r$ encounter a robot $s$ in some round $1 \leq k' < k$ that decides to settle in that round. If $s$ is a non-Byzantine robot, then $s$ never changes its position on the ring and $r$ can only encounter $s$ again in some round $k' + n > n$ (since $r$ only moves in one direction in the ring). Hence, if $r$ finds $s$ in another node, i.e., at some round after $k'$ and before $k'+ n$, then $s$ must be a Byzantine robot. In other words, if $s \in M^k_r \bigcap (G'^k_r \bigcup G^k_r)$ then $r$ can identify $s$ as a Byzantine robot.
\end{proof}

\begin{lemma}\label{lemma:non-Byzantine}
If $r$ does not settle down at node $v_k$ in round $k$, then there is at least one robot that $r$ currently does not consider Byzantine that is written in $A_r[k]$.
\end{lemma}

\begin{proof}
Recall that in round $k$, a robot $s$ that is co-located with $r$ is considered by $r$ to be non-Byzantine iff $s \notin B^k_r$.

Robot $r$ does not settle at node $v_k$ if either of the two checks in Line~\ref{item:line 1} and Line~\ref{item:line 2} succeed. Both checks require that a robot $s \neq r$ such that $s \notin B^k_r$ chooses to settle at $v_k$. Hence, if either of the two checks succeed, then by definition $r$ considers at least one of the robots that settled to be non-Byzantine.
\end{proof}

Note that a robot that appears non-Byzantine to $r$ may in fact be a Byzantine robot. In case $r$ adds more than one robot to $A_r[k]$ for a given round $k$, we desire that at most one of the robots in $A_r[k]$ may be non-Byzantine. It may in fact be the case that all of the robots in $A_r[k]$ are Byzantine. However, it should not be the case that more than one robot in $A_r[k]$ is non-Byzantine, i.e., more than one non-Byzantine robot settles at the same node. The following lemma shows that this is indeed the case.

\begin{lemma}\label{lemma:singlesettles}
No two non-Byzantine robots settle at the same node $v_k$ in any round $k$.
\end{lemma}
\begin{proof}
Let $r$ and $s$ be two non-Byzantine robots co-located on node $v_k$. If one of them is already settled there, say $s$ without loss of generality, then $r$ will not settle on the node because of the check in Line~\ref{item:line 1}. This is because if $s$ is non-Byzantine, then it would not appear in $B^k_r$.

If both $r$ and $s$ are in $G^k_r$, i.e., both are not yet settled at the node, then we show that it is impossible for both of them to settle at $v_k$. Without loss of generality, let the ID of $s$ be less than that of $r$. If $s$ passes one of the two checks in Line~\ref{item:line 1} and Line~\ref{item:line 2} and thus does not settle at $v_k$, we do not need to show anything further. However, if $s$ chooses to settle down at node $v_k$, then it is guaranteed that the check in Line~\ref{item:line 2} will succeed for robot $r$, and thus $r$ will not settle at $v_k$.
\end{proof}

Now we are ready to prove Theorem \ref{theorem:nrounds}.
\begin{proof}[Proof of Theorem \ref{theorem:nrounds}]
Let $r$ be a non-Byzantine robot. By Lemma~\ref{lemma:non-Byzantine}, we see that in each round $k$, $1 \leq k \leq n$, if $r$ is not yet settled, then at least one non-Byzantine robot's ID (from $r$'s perspective) is written in $A_r[k]$. So, by the end of round $n$, $r$ definitely found a node to settle at as there are only $n$ robots in total. This is true for all non-Byzantine robots. Thus after $n$ rounds, each non-Byzantine robot has settled at some node.

By Lemma~\ref{lemma:singlesettles}, no two non-Byzantine robots settle at the same node. Thus, after $n$ rounds, Byzantine dispersion is solved.
\end{proof}

\section{Conclusions and Future Work}
\label{sec:conclusion}
To recap, we have presented a time-memory optimal deterministic algorithm to solve Byzantine dispersion on a ring. We then presented two deterministic algorithms which were either only time optimal or only memory, but allowed us to solve Byzantine dispersion using less assumptions. Additionally, we developed a useful primitive to achieve Byzantine dispersion, when all non-Byzantine robots are initially present on the same node, in a time-memory optimal manner and requiring no assumptions.

A number of interesting research directions that result from this paper. An interesting open problem arising from our work is that of reducing assumptions required to achieve dispersion. Specifically, our time and memory optimal algorithm required the following four assumptions: (i) the robots' unique IDs are taken from the range $[1,n]$, (ii) each robot knows the value of $f$, (iii) $f \leq \lfloor (n-4)/17 \rfloor$, and (iv) the \textit{follow} primitive held. Is it possible to develop a time and memory optimal algorithm that drops one or all of these assumptions? This paper only focused on solutions of the problem for the ring. An exciting line of research is to study this problem on other types of graphs and eventually develop algorithms that are optimal for any graph. 
Another type of generalization relates to time; an understanding of how solutions to Byzantine dispersion look in the asynchronous system warrants study.

There is also an interesting line of research available to pursue. As mentioned in the related work, there is not much literature on exploration in the presence of faulty robots. Since dispersion of $n$ robots on a $n$-node graph also requires exploration of all the nodes, intuitively, solutions to dispersion should readily lend themselves to exploration. Furthermore, dispersion is at least as hard as exploration and possibly harder. Thus, we believe that this paper might open the doors to new and interesting solutions for exploration in the presence of faulty robots.
\bibliographystyle{plain}
\bibliography{ref}
\end{document}